\documentclass[11pt]{article}
\usepackage[latin1]{inputenc}
\usepackage{fullpage}  % sets all 4 margins to be either 1 inch or 1.5 cm, and specifies the page style

\usepackage{amsmath}   % enhances the typeset appearance of mathematical formulas
\usepackage{amsfonts}  % more fonts
\usepackage{amssymb}   % for more math symbols
\usepackage{amsthm}    % for theorem-like environments

\usepackage{algorithm}
\usepackage{algorithmic}

\newtheorem{theorem}{Theorem}

\newtheorem{lemma}{Lemma}

\newcommand{\BO}[1]{{O}\left(#1\right)}

\newcommand{\BT}[1]{{\Theta}\left(#1\right)}
\newcommand{\BOM}[1]{{\Omega}\left(#1\right)}

\newcommand{\A}{\mathcal{A}}
\newcommand{\e}{\epsilon}
\newcommand{\p}{P}
\newcommand{\de}{\delta}

\title{Communication Lower Bounds for Distributed-Memory Computations\thanks{This
	work was supported, in part, by the University of Padova Projects \emph{STPD08JA32}
	and \emph{CPDA121378}.}}
	
\author{Michele Scquizzato\thanks{Department of Computer Science, University of
   Pittsburgh, Pittsburgh, PA 15260, USA. Supported by a fellowship
   of ``Fondazione Ing.\ Aldo Gini'', University of Padova, Italy.
   \hbox{E-mail}:~\texttt{scquizza@pitt.edu}.
   Most of this work was done while this author was a Ph.D.\ student
   at the University of Padova.}
   \and
   Francesco Silvestri\thanks{Department of Information Engineering, University of
   Padova, 35131 Padova, Italy. 
   \hbox{E-mail}:~\texttt{silvest1@dei.unipd.it}}}

\date{\today}

\begin{document}

\maketitle

\begin{abstract}
We give lower bounds on the communication complexity required to solve several
computational problems in a distributed-memory parallel machine, namely standard
matrix multiplication, stencil computations, comparison sorting, and the Fast
Fourier Transform. We revisit the assumptions under which preceding results were
derived and provide new lower bounds which use much weaker and appropriate hypotheses.
Our bounds rely on a mild assumption on work distribution, and strengthen previous
results which require either the computation to be balanced among the processors,
or specific initial distributions of the input data, or an upper bound on the size
of processors' local memories.
\end{abstract}

\noindent \textbf{Keywords}: Communication, lower bounds, distributed memory,
parallel algorithms, BSP.

\newpage

\section{Introduction}\label{sec:intro}

Communication is a major factor determining the performance of algorithms
on current computing systems, as the time and energy needed to transfer
data between processing and storage elements is often significantly higher
than that of performing arithmetic operations. The gap between computation
and communication costs, which is ultimately due to basic physical principles,
is expected to become wider and wider as architectural advances allow to build
systems of increasing size and complexity. Hence, the cost of data movement
will play an even greater role in future years.

As in all endeavors where performance is systematically pursued, it is
important to evaluate the distance from optimality of a proposed algorithmic
solution, by establishing appropriate lower bounds. Given the well-known
difficulty of establishing lower bounds, results are often obtained under
restrictive assumptions that may severely limit their applicability. It is
therefore important to progressively reduce or fully eliminate such restrictions.

In this spirit, we consider lower bounds on the amount of communication that is
required to solve some classical computational problems on a distributed-memory
parallel system. Specifically, we revisit the assumptions and constraints
under which preceding results were derived, and prove new lower bounds which
use much weaker hypotheses and thus have wider applicability. Even when the
functional form of the bounds remains the same, our results do yield new
insights to algorithm developers since they might reveal if some settings
are needed, or not, in order to obtain better performance.

We model the machine using the standard \emph{Bulk Synchronous Parallel}
(BSP) model of computation~\cite{Valiant90}, which consists of a collection
of $p$ processors, each equipped with an unbounded private memory and
communicating with each other through a communication network. The starting
point of any investigation of algorithms for a distributed-memory model is
the specification of an I/O protocol, which defines where input elements reside
at the beginning of the computation and where the outputs produced by the
algorithm must be placed. The distribution of inputs and outputs effectively
forms a part of the problem specification, thus restricting the applicability
of upper and lower bounds. Much of previous work on BSP algorithms considers a
version of the BSP model equipped with an additional \emph{external memory},
which serves as the source of the input and the destination for the output
(see, e.g.,~\cite{Tiskin11}). This modification significantly alters the spirit
of the BSP of serving as a model for distributed-memory machines, making
it very similar to shared-memory models like the LPRAM~\cite{AggarwalCS90}.
In fact, in a distributed-memory machine, the inputs might already be distributed
in some manner prior to the invocation of the algorithm, and the outputs are
usually left distributed in the processors' local memories at the end of the
execution, especially if the computation is a subroutine of a larger computation.
Thus, lower bounds that use this assumption, which essentially exploit this
``hack'' to guarantee that acquiring the $n$ input elements contributes
to the communication cost of algorithms (as some processor must read at least
$\lceil n/p \rceil$ input values), are not directly applicable to distributed-memory
architectures.

Other authors, within the original BSP model, assume specific distributions
of the input data. As we shall see later, it is usually assumed that the input
is initially \emph{evenly} distributed among the $p$ processors. However,
this apparently ``reasonable'' hypothesis is not part of the logic of the BSP
model. In fact, the physical distribution of input data across the processors
may depend on several factors, ranging from how the inputs get acquired to the
file system policies. Moreover, this hypothesis may lead to unsatisfactory
communication lower bounds. Consider, e.g., the computation of a directed
acyclic graph (DAG) with ``few'' input nodes and a ``long'' critical path.
In this case, naive algorithms which entrust the whole computation to one
processor might be communication optimal. This is misleading, since it
steers towards algorithms which are not parallel at all.

One possibility to overcome both the issues discussed above is to require,
in place of the even distribution of the inputs and of the presence of an
external memory, that algorithms exhibit some level of \emph{load balancing}
of the computation. Typically, if $\mathcal{W}$ denotes the total work
required by any algorithm to solve the given problem, it is required that each
processor performs $\BO{\mathcal{W}/p}$ elementary computations. However,
this way we are assuming, but not proving, that optimal solutions balance
computation. In fact, in general there is a tradeoff between computation costs
and communication costs. Some papers (see, e.g.,~\cite{PapadimitriouU87,WuK91})
quantify such tradeoffs by establishing lower bounds on the communication cost
of any algorithm as a function of its computation time. Nevertheless, results
of this kind usually indicate that the higher lower bounds on communication
correspond only to perfectly (to within constant factors) work-balanced computations,
and such bounds are tight since achieved by balanced algorithms. This leaves open
the possibility that a substantial saving on communication costs could actually
be achieved at a price of a small unbalance of the computation loads. 

Another common assumption is putting an upper bound on the \emph{size} of
processors' local memories. However, current technological advances allow
to build cheap memory and storage devices that, for many applications, allow
a single machine to store the whole input data set and the intermediate data.
Moreover, results derived under this assumption are less general than results
that put no limits on the amount of storage available to processors; indeed,
lower bounds are relatively easier to establish, as the model essentially
becomes a parallel version of the standard external memory (EM) model for
sequential computations, for which much more results and techniques are
known (see, e.g.,~\cite{HongK81,AggarwalV88}).

In contrast, lower bounds presented in this paper do not hinge on any of
the above assumptions. We develop new lower bounds for a number of key
computational problems, namely standard matrix multiplication, stencil
computations, comparison sorting, and the Fast Fourier Transform, using
the weak assumption that no processor performs more than a constant
fraction of the total required work. This requires more involved arguments,
and substantially strengthen previous work on communication lower bounds
for distributed-memory computations.

\paragraph*{The model.}
The \emph{Bulk Synchronous Parallel} (BSP) model of computation was introduced
by Valiant~\cite{Valiant90} as a bridging model for general-purpose parallel
computing. The architectural component of the model consists of $p$ processing
elements $P_0,P_1,\dots,P_{p-1}$, each equipped with an unbounded local memory,
interconnected by a communication medium. The execution of a BSP algorithm
consists of a sequence of \emph{supersteps}, where each processor can perform
operations on data in its local memory, send messages and, at the end, execute
a global synchronization. The running time of the $i$-th superstep is expressed
in terms of two parameters, $g$ and $\ell$, as $T_i = w_i + h_i g + \ell$, where
$w_i$ is the maximum number of local operations performed by any processor, and
$h_i$ is the maximum number of messages sent or received by any processor.
The \emph{running time} $T_{\mathcal A}$ of a BSP algorithm $\mathcal A$
is the sum of the times of its supersteps and can be expressed as
$W_{\mathcal A} + H_{\mathcal A}g + S_{\mathcal A}\ell$, where $S_{\mathcal A}$
is the number of supersteps, $W_{\mathcal A} = \sum_{i=1}^{S_{\mathcal A}} w_i$
is the \emph{local computation complexity}, and $H_{\mathcal A} =
\sum_{i=1}^{S_{\mathcal A}} h_i$ is the \emph{communication complexity}.

\paragraph*{Previous work.}
The complexity of communication on various models of computation has received
considerable attention. Lower bounds are often established through adaptations
of the techniques of Hong and Kung~\cite{HongK81} for hierarchical memory,
or by critical path arguments, such as those in~\cite{AggarwalCS90}.
For applications of these and other techniques
see~\cite{PapadimitriouU87,AggarwalV88,Savage95,Goodrich99,BilardiP99,BilardiPD00,IronyTT04,RanjanSZ11,BallardDHS11,BilardiSS12,BallardDHS12}
as well as~\cite{Savage98} and references therein. In the following, we discuss previous
work on lower bounds for the communication complexity of the problems studied in
this paper.

A standard computational problem is the multiplication of two $n \times n$
matrices. For the classical $\BT{n^3}$ algorithm, an $\BOM{n^2/p^{2/3}}$ lower
bound has been previously derived for the BSP~\cite{Tiskin98} and the
LPRAM~\cite{AggarwalCS90}. However, both results hinge on the hypothesis that
the input initially resides outside the processors' local memories and thus
must be read, contributing to the communication complexity of the algorithms.
As such, these results are an immediate consequence of a result of~\cite{HongK81}
(then restated in~\cite{IronyTT04}) which, loosely speaking, bounds from above
the amount of computation that can be performed with a given quantity of data.
When input is assumed to be initially evenly distributed across the $p$ processors'
local memories, the same lower bound is claimed in~\cite{CheathamFSV95}. Recently,
Ballard et al.~\cite{BallardDHLS12} obtained a result of the same form by assuming
perfectly balanced (to within constant factors) computations, and disallowing any
initial replication of inputs. Restricting to balanced computations allows to reduce
to the situation where inputs are evenly distributed: in fact, it is easy to prove
that, given a problem on $n$ inputs and which requires $N$ operations to be solved,
if each processor performs at most $\alpha N/p$ operations, for some $\alpha$ with
$1 \leq \alpha \leq p/2$, then there exists one processor which initially holds at
most $2\alpha n/p$ inputs, and that performs at least $\alpha/(2\alpha-1) \cdot N/p$
operations. The very same bound was found also by Irony et al.~\cite{IronyTT04},
who restrict their attention to computations that take place on machines
where processors' local memory size is assumed to be $M = \BO{n^2/p^{2/3}}$.
Finally, Solomonik and Demmel~\cite{SolomonikD11} investigate tradeoffs between
input replication and communication complexity (see also~\cite{BallardDHS11}).

A class of computations ubiquitous in scientific computing is that of stencil
computations, where each computing node in a multi-dimensional grid is updated
with weighted values contributed by neighboring nodes. These computations include
the \emph{diamond} DAG in the two-dimensional case and the \emph{cube} DAG in
three dimensions. For the former, Papadimitriou and Ullman~\cite{PapadimitriouU87}
present a communication-time tradeoff which yields a tight $\BOM{n}$ lower bound
on the communication complexity only for the case of balanced computations.
Aggarwal et al.~\cite{AggarwalCS90} extend this result to all algorithms whose
computational complexity is within a constant factor of the number of nodes of
the DAG. To the best of our knowledge, this is the sole example of a tight lower
bound that holds under the same hypothesis used in this paper. By generalizing
the technique in~\cite{PapadimitriouU87}, Tiskin~\cite{Tiskin98} establishes a
tight bound for the cube DAG, and claims its extension to higher dimensions.
However, this results only hold when the computational load is balanced among
the $p$ processors.

Another key problem is sorting. Many papers assume that the $n$ inputs initially
reside outside processors' local memories, thus obtaining an $\BOM{n/p}$ lower
bound which turns out to be tight when it is additionally assumed that problem
instances have sufficient \emph{slackness}, that is, $n >> p$ (e.g., $p^2 \leq n$
is a common assumption). Under some technical assumptions, a bound of the form
$\BOM{n \log n/ (p \log(n/p))}$, which is tight for all values of $p \leq n$,
was first given within the LPRAM model~\cite{AggarwalCS90}.\footnote{For
notational convenience, within asymptotic notations we will henceforth use
$\log x$ to denote $\max \{1,\log_2 x\}$.} This bound, however, includes the
cost to read the input from the shared memory. A similar lower bound was derived
later by Goodrich~\cite{Goodrich99} within the BSP model, but the result holds
only for the subclass of algorithms performing supersteps of degree $h = \BT{n/p}$,
and when the inputs are evenly distributed among the processors.

Previous work on the communication required to compute an FFT DAG of size $n$
is similar to previous work for sorting. By exploiting the property that, as
shown in~\cite{WuF81}, the cascade of three FFT networks has the topology of
a full sorting network, the aforementioned lower bounds for sorting also hold
for the FFT DAG. In a recent paper~\cite{BilardiSS12}, we obtain the same result
assuming that the maximum number of outputs held by any processor at the end of
the algorithm is at most $n/2$, and without assumptions on the distribution of
the input and of the computational loads; while these hypotheses are not equivalent
to the one we are using in this paper, the result in~\cite{BilardiSS12} is the
closest to the one that we will develop in Section~\ref{sec:FFT}.

\paragraph*{Our contribution.}

In this paper we present lower bounds on the communication complexity
required by key computational problems such as standard matrix multiplication,
stencil computations, comparison sorting, and the Fast Fourier Transform,
when solved by parallel algorithms on the BSP model. These results, which
are all tight for the whole range of model parameters, rely solely on the hypothesis
that no processor performs more than a \emph{constant} fraction of the total
required work. More formally, let $\mathcal{W}$ be the total work required by
any algorithm to solve the given problem (if the problem is represented by a
directed acyclic graph, then $\mathcal{W}$ is the number of nodes of the DAG,
otherwise $\mathcal{W}$ is a lower bound on the computation time required by
any sequential algorithm), and let $W$ be the maximum amount of work performed
by any BSP processor; then, in the same spirit of the aforementioned result
for the diamond DAG in~\cite{AggarwalCS90}, $W$ is assumed to satisfy the
bound $W \leq \e \mathcal{W}$, for some constant $\e \in (0,1)$. The rationale
behind this approach is that communication is the major bottleneck of a
distributed-memory computation unless the latter is sequential or ``nearly
sequential'', in which case the main contribution to the running time $T$ of
an algorithm comes from computation. Since it is directly linked to the running
time metric, and it does not allow for any other restrictive assumptions suggested
by orthogonal constraints, we believe that this is the right approach to perform
a systematic analysis of the communication requirements of distributed-memory
computations.

We emphasize that, in contrast to previous work, our lower bounds do not count the
communication required to acquire the input, allow for any initial distribution
of the input among the processors' local memories, assume no upper bound on the
sizes of the latter, and do not require computations to be balanced. On the other
hand, some of our results make use of additional technical assumptions, such as
the non-recomputation of intermediate results in the course of the computation,
or some restrictions on the replication of input data. Such restrictions, however,
were already in place in almost all of the corresponding state-of-the-art lower
bounds.

\section{Matrix Multiplication}\label{sec:mmult}

In this section we consider the problem of multiplying two $n \times n$ matrices,
$A$ and $B$, using only semiring operations, that is, addition and multiplication.
Hence, each element $c_{i,j}$ of the output matrix $C$ is an explicit sum of
products $a_{i,k} \cdot b_{k,j}$, which are called \emph{multiplicative terms}.
This rules out, e.g., Strassen's algorithm~\cite{Strassen69} and the Boolean
matrix multiplication algorithm of Tiskin~\cite{Tiskin98a}. As shown in~\cite{Kerr70},
any algorithm using only semiring operations must compute at least $n^{3}$ distinct
multiplicative terms.

In this section we establish a lower bound on the communication complexity
of any parallel algorithm for matrix multiplication on a BSP with $p$ processors.
This result is derived assuming that no processor performs more than a
constant fraction of the $n^{3}$ total work  required by any algorithm,
measured as the number of scalar multiplications, and that each input element is
initially stored in the local memory of exactly one processor. The bound has the
form of $\BOM{W^{2/3}}$, where $W$ is the maximum number of multiplicative terms
evaluated by a processor, and is tight for all values of $p$ between two and $n^2$.
The argument trough which we establish such a result is a repeated application
of a ``bandwidth'' argument which, loosely speaking, is as follows. Consider a
processor which performs the maximum amount of work. If this processor initially
holds ``few'' input values, then, since it computes at least $n^{3}/p$ multiplicative
terms, it must receive ``many'' inputs from the submachine including the other
processors; otherwise, if it initially holds ``many'' inputs, then it has to send
many of them to the other processors, because it cannot perform too much work on
its own, and thus the other processors have to perform at least a constant fraction
of the total work. The lower bound applies to any distribution of input and output 
matrices, and only requires that the input matrices are not initially replicated.

Towards this end, we first establish a lower bound of $\BOM{n^2}$ under the
same hypotheses outlined above for two processors. This result is derived using
a bandwidth argument that bounds from below the amount of data that must travel
across the communication network of a two-processor machine. A bound of the same
form can be found in~\cite[Section~6]{IronyTT04}, which holds only when the
elements of the input matrices $A$ and $B$ are evenly, or almost evenly, distributed
among the two processors. Our result, which instead allows any initial distribution
of the input matrices (without replication), establishes the same bound by using
a mild hypothesis on the maximum computation load faced by the processors.

\begin{lemma}\label{lem:bandwidth}
Let $\A$ be any algorithm for computing the matrix product $C=AB$, using only semiring
operations, on a BSP with two processors. If each processor computes at most $\epsilon
n^{3}$ multiplicative terms, where $\epsilon$ is an arbitrary constant in $(1/2,1)$,
and the input matrices are not initially replicated, then the communication complexity
of the algorithm is 
\[
H_{\A}(n,p) = \BOM{n^2}.
\]
\end{lemma}

\begin{proof}
We use a bandwidth argument as the one employed in~\cite[Theorem~6.1]{IronyTT04}.
By hypothesis, each processor computes at most $\e n^{3}$ multiplicative
terms. Let $K$ be the number of elements of $C$ whose corresponding multiplicative
terms have \emph{not} been totally computed by the same processors. If $K \geq
(1-\e)n^2/2$, then the communication complexity is at least $\BOM{n^2}$ since a
processor receives a message for at least $K/2$ of such entries containing a
multiplicative term or a partial prefix sum.

Suppose now that $K < (1-\e)n^2/2$. Then there are at least $(1+\e)n^2/2$
entries of $C$ whose $n$ multiplicative terms have been entirely computed by
the same processor. We denote with $n_0$ and $n_1$ the number of entries of
$C$ computed entirely by processor $\p_0$ and $\p_1$, respectively, and suppose
without loss of generality that $n_0 \geq n_1$. Clearly, $n_0+n_1 \geq (1+\e)n^2/2$.
Since $n_0 \geq n_1$ and since each processor can compute at most $\e n^{3}$
multiplicative terms, it follows that $n_0 \leq \e n^2$, and thus $n_1\geq (1-\e)n^2/2$.
Let $r_i$ and $c_i$ denote the number of rows of $A$ and columns of $B$, respectively,
whose ${n}$ entries have all been accessed by processor $\p_i$, with $i \in \{0,1\}$,
during the lifespan of the algorithm. Since $n_i$ entries of $C$ are computed entirely
by processor $\p_i$, then we have $r_i c_i \geq n_i$. Let $\alpha = \sqrt{\e}+\sqrt{(1-\e)/2}$;
we observe that $\alpha\in(1/2+1/\sqrt{2},1)$ since $\epsilon \in (1/2,1)$. If
$r_0+r_1 \geq \alpha {n}$, at least $(\alpha-1)n = \mu n$ rows of $A$, where $\mu$ is a
suitable constant in $(0,1/\sqrt{2}-1/2)$, are used by both processors, incurring $\BOM{n^2}$
messages for exchanging the rows since, by hypothesis, the input matrices are not initially
replicated. Suppose now that $r_0 + r_1 < \alpha {n}$. Then, we have 
\[
c_0 + c_1 \geq \frac{n_0}{r_0} + \frac{n_1}{r_1} \geq 
\frac{n_0}{r_0} + \frac{n_1}{\alpha {n}-r_0}.
\]
By taking the derivative of the last term with respect to $r_0$ we see that the
right-hand side is minimized at $r_0 = {\alpha n \sqrt{n_0}}/{(\sqrt{n_0}+\sqrt{n_1})}$,
whence
\[
c_0 + c_1 \geq \frac{\left(\sqrt{n_0}+\sqrt{n_1}\right)^2}{\alpha {n}}.
\]
Since $n_0+n_1 \geq (1+\e)n^2/2$ and $n_0 \leq \e n^2$, the term
$\sqrt{n_0} + \sqrt{n_1}$ is minimized when $n_0$ assumes the largest
allowed value (i.e., $\e n^2$) and $n_0+n_1 = (1+\e)n^2/2$. Thus, we have
\[
c_0 + c_1 \geq \frac{\left(\sqrt{\e n^2}+\sqrt{(1-\e)n^2/2}\right)^2}{\alpha n} = \alpha{n}.
\]

The lemma follows since $(\alpha-1){n}=\BT{{n}}$ columns of $B$ are used by both
processors, entailing $\BOM{n^2}$ messages for exchanging them.
\end{proof}

To prove our main result we also need the following technical lemma, which was
first given by Hong and Kung in their seminal paper on I/O complexity, and then
restated in~\cite{IronyTT04} by applying the Loomis-Whitney inequality~\cite{LoomisW49}.

\begin{lemma}[{\cite[Lemma~6.1]{HongK81}}]\label{lem:mmult}
Consider the matrix multiplication of two $n \times n$ matrices $A$ and $B$
using scalar additions and multiplications only. During the computation,
if a processor accesses at most $K$ elements of each input matrices and
contributes to at most $K$ elements of the output matrix $C$, then it can
compute at most $2K^{3/2}$ multiplicative terms.
\end{lemma}

Now we have all the tools to prove the main result of this section. The following
theorem establishes an $\BOM{W^{2/3}}$ lower bound to the communication complexity
of any standard algorithm, where $W$ denotes the maximum number of multiplicative
terms evaluated by a processor. By the result of~\cite{Kerr70} and by the pigeonhole
principle, there exists a processor that computes at least $n^3/p$ multiplicative
terms, from which the standard $\BOM{n^2/p^{2/3}}$ lower bound follows.

\begin{theorem}\label{thm:mm-lb}
Let $\A$ be any algorithm for computing the matrix product $C = AB$, using only
semiring operations, on a BSP with $p$ processors, where $1 < p \leq n^2$, and
let $W$ be the maximum number of multiplicative terms evaluated by a processor.
If $W \leq \max\{n^3/p, n^3/11^3\}$, and the input matrices are not initially
replicated, then the communication complexity of the algorithm is
\[
H_{\A}(n,p) = \BOM{W^{2/3}}.
\]
\end{theorem}

\begin{proof}
Without loss of generality, we assume that any multiplicative term computed by
the processors is actually used towards the computation of some entry of the
output matrix $C$ (that is, processors do not perform ``useless'' computations).
Consider one of the processors that compute $W$ multiplicative terms, and without
loss of generality let $\p_0$ denote such a processor. Let $I$ be the number
of input elements initially held by this processor in its local memory.

Consider first the case $I \leq W^{2/3}/5$. By Lemma~\ref{lem:mmult}, a
processor that computes $W$ multiplicative terms either accesses, during the
whole execution of algorithm $\A$, at least $(W/2)^{2/3}$ input elements, or
computes multiplicative terms relative to at least $(W/2)^{2/3}$ elements of
the output matrix. In the first case, since $\p_0$ initially holds $I \leq W^{2/3}/5$
input elements, it must receive at least $(W/2)^{2/3} - I = \BOM{W^{2/3}}$ data
words from other processors,  and the theorem follows. On the other hand, suppose
$\p_0$ computes multiplicative terms relative to $(W/2)^{2/3}$ entries of the output
matrix, and partition such entries into three groups: $G_1$, the set of entries whose
multiplicative terms have all been computed by the processor; $G_2$, the set of
entries produced by the processor but for which some multiplicative term or partial
sum has been communicated by some other processor; $G_3$, the set of entries not
produced by the processor. Clearly, at least one of these three groups must have
size at least $(W/2)^{2/3}/3$. If $\vert G_1 \vert \geq (W/2)^{2/3}/3$, then
$\p_0$ must have computed at least ${n}(W/2)^{2/3}/3$ multiplicative terms,
and since any entry of the input matrices occurs in only ${n}$ of such terms,
the processor must have received $(W/2)^{2/3}/3 - I = \BOM{W^{2/3}}$ elements
from other processors. If $\vert G_2 \vert \geq (W/2)^{2/3}/3$, then for each
entry in $G_2$ $\p_0$ has received some term from other processors, therefore
accounting for a total of $\BOM{W^{2/3}}$ incoming data words. Finally, if
$\vert G_3 \vert \geq (W/2)^{2/3}/3$, then, since any multiplicative term must
be used towards the computation of some entry of the output matrix $C$, for each
entry in $G_3$, $\p_0$ must send some multiplicative term or partial sum to the
processor that will produce the corresponding entry of $C$, and this implies that
$\p_0$ must send $\BOM{W^{2/3}}$ data words. In all three cases, $\BOM{W^{2/3}}$
messages have to be exchanged by $\p_0$ with the other processors, and the claim
follows for $I \leq W^{2/3}/5$.

Now suppose $I > W^{2/3}/5$ and $p \geq 11^3$. Assume, without loss of generality,
that $\p_0$ initially holds at least $I/2$ elements of matrix $A$. Since any entry
of the input matrices occurs in ${n}$ multiplicative terms, there are at least
$In/2$ multiplicative terms that depend on the entries of $A$ initially held
by the processor. Since $W$ multiplicative terms are computed by the processor,
the remaining $In/2 - W$ ones are computed by other processors. Since, by hypothesis,
each entry of $A$ is initially non replicated and a processor can compute at most
$n$ multiplicative terms using a single entry of $A$, we have that $(In/2 - W)/n$
messages are required for sending the appropriate entries of $A$ to the processors
that will compute the remaining entries. Hence, $H_{\A}(n,p) \geq (In/2 - W)/n$.
Finally, observe that since $p \geq 11^3$, then by hypothesis it holds that
$W \leq n^{3}/11^3$. Putting all pieces together yields
\begin{align*}
H_{\A}(n,p) &\geq \frac{I n/2 - W}{n} \\
&> \frac{W^{2/3}}{10} - \frac{W}{n} \\
&= W^{2/3}\left(\frac{1}{10} - \frac{W^{1/3}}{n}\right) \\
&\geq W^{2/3}\left(\frac{1}{10} - \frac{1}{11}\right) \\
&= \frac{W^{2/3}}{110},
\end{align*}
which concludes the proof of the second case.

Finally, when $I > W^{2/3}/5$ and $p < 11^3$, the sought lower bound follows by
Lemma~\ref{lem:bandwidth}. Indeed, the $p$ processors can be virtually partitioned
into two subsets, each consisting of exactly $p/2$ processors; in particular,
processor $\p^*_0$ will be identified with the submachine including the first
half of the $p$ processors, and $\p^*_1$ with the submachine including the second
half. Since $p < 11^3$, by hypothesis each BSP processor computes at most $n^{3}/p$
multiplicative terms, and thus both $\p^*_0$ and $\p^*_1$ compute at most
$(n^{3}/p) (p/2) = n^{3}/2$ multiplicative terms overall. Hence we can apply 
Lemma~\ref{lem:bandwidth} to processors $\p^*_0$ and $\p^*_1$, obtaining the
desired result.
\end{proof}

The proposed bound is tight and is matched by the algorithm that decomposes the
problem into $n^3/W \leq p$ subproblems of size $W^{1/3} \times W^{1/3}$, and then
solves each subproblem sequentially in each round. Since $W \geq n^3/p$, the
minimum communication complexity is $\BOM{n^2/p^{2/3}}$, which is achieved by
the standard 3D algorithm~\cite{IronyTT04}.

Finally, we observe that the above theorem can be extended to the case $W \leq \e n^3$,
for an arbitrary constant $\e\in(0,1)$, as soon as each multiplicative term is computed
once. Also, we remark that, if each processor holds $\BO{W^{2/3}}$ inputs, our bound
applies even when each input element may be present in more than one processor at the
beginning of the computation. We also conjecture that our bound can be extended up
to $p^{1/3}$ replication, as shown in~\cite{SolomonikD11} assuming balanced memory
or work.

\section{Stencil Computations}\label{sec:stencil}

A \emph{stencil} defines the computation of an element in a $d-1$-dimensional
spatial grid at time $t$ as a function of neighboring grid elements at time
$t-1,\dots,t-\tau$, for some value $\tau \geq 1$ and constant $d > 1$ (see,
e.g.,~\cite{FrigoS05}). We provide an $\BOM{n^{d-1}/p^{(d-2)/(d-1)}}$ lower
bound to the communication complexity of any algorithm evaluating $n$ time
steps of a $d-1$-dimensional stencil. For simplicity we assume $\tau = 1$,
however our bounds still apply in the general case. The bound follows by
investigating the \emph{$(n,d)$-array problem}, which consists in evaluating
all nodes of a $d$-dimensional array DAG of size $n$. Indeed, the DAG given by the
$(d-1)$-dimensional grid plus the time dimension spans a $d$-dimensional spacetime
containing an $(n/2,d)$-array as a subgraph. A \emph{$d$-dimensional array DAG} has
$n^{d}$ nodes $\langle i_0,\dots,i_{d-1} \rangle$, for each $0 \leq i_0,\dots,i_{d-1}
< n$, and there is an arc from $\langle i_0,\dots,i_k,\dots,i_{d-1} \rangle$ to
$\langle i_0,\dots,i_k+1,\dots,i_{d-1} \rangle$, for each $0 \leq k < d$ and
$0 \leq i_0,\dots,i_{d-1} < n-1$. Observe that $\langle 0,\dots,0 \rangle$ and
$\langle n-1,\dots,n-1 \rangle$ are the single input and output nodes, respectively. 

Our result hinges on the restriction on the nature of the computation whereby each
vertex of the DAG is computed exactly once. In this setting, the crucial property
is that for each arc $(u,v)$ such that $u$ is computed by processor $P$ and $v$
is computed by processor $P'$, $P \neq P'$, there corresponds a message from $P$
to $P'$ (which may also cross other processors). Such arcs are referred to as
\emph{communication arcs}.

We now introduce some preliminary definitions, which will be used throughout
the section. We envision an $(n,d)$-array as partitioned into $p^{d/(d-1)}$ smaller
$d$-dimensional arrays, called \emph{blocks}, of size $n/p^{1/(d-1)}$, and
denote each block with $B_{i_0,\dots,i_{d-1}}$ for $0 \leq i_0,\dots,i_{d-1}
< p^{1/(d-1)}$. Block $B_{i_0,\dots,i_{d-1}}$ contains nodes
$\langle i'_0,\dots,i'_{d-1} \rangle$, for each $i_k n/p^{1/(d-1)} \leq i'_k <
(i_k+1) n/p^{1/(d-1)}$. A block has $n^d/p^{d/(d-1)}$ nodes, and is said
\emph{$\ell$-owned} if more than half of its nodes are evaluated by processor
$\p_\ell$, with $0 \leq \ell < p$. A block is \emph{owned} if there exists some
$\ell$, with $0 \leq \ell < p$, such that it is $\ell$-owned; it is \emph{shared}
otherwise. Two blocks $B_{i_0,\dots,i_{d-1}}$ and $B_{i'_0,\dots,i'_{d-1}}$
are said to be \emph{adjacent} if their coordinates differ in just one position
$k$ and $\vert i_k - i'_k \vert =1$ (i.e., they share a face).
For the sake of simplicity, we assume that $n$ and $p$ are powers of $2^{d-1}$ and thus the
previous values (e.g., $n/p^{1/(d-1)}$) are integral: since $d$ is a constant, this assumption is
verified by suitably increasing $n$ and decreasing $p$ by a constant factor which does not
asymptotically affect our lower bounds.
 
In order to establish our main lower bound, we need two preliminary lemmas.
The first one gives a slack lower bound based on the $d$-dimensional version
of the Loomis-Whitney geometric inequality~\cite{LoomisW49}, and reminds the
result of Theorem~\ref{thm:mm-lb} for matrix multiplication when $d=3$.

\begin{lemma}\label{lem:w}
Let $\A_d$ be any algorithm solving the $(n,d)$-array problem, without recomputation,
on a BSP with $p$ processors, where $1 < p \leq n^{d-1}$, and denote with $W$ the maximum
number of nodes evaluated by a processor. If $W \leq \epsilon n^d$, for an arbitrary
constant $\epsilon \in (0,1)$, then the communication complexity of the algorithm is
\[
H_{\A_d}(n,p) = \BOM{W^{(d-1)/d}}.
\]
\end{lemma}

\begin{proof}
Let $\p_0$ be a processor evaluating $W$ nodes, and suppose $W \leq n^d/2$.
Denote with $\Psi$ the set of nodes evaluated by $\p_0$, and with $N_i$ the
set of points obtained by dropping the $i$-th dimension from the set $\Psi$,
for each $0 \leq i < d$. Let $n_i = \vert N_i \vert$. Applying the discrete
Loomis-Whitney inequality~\cite{LoomisW49}, we have $W^{d-1} \leq \Pi_{i=0}^{d-1} n_i$,
and hence $\max\{n_0,\dots,n_{d-1}\} \geq W^{(d-1)/d}$. Assume, without loss of
generality, $n_0 \geq W^{(d-1)/d}$. Let the set $N'_0$ contain the points in
$N_0$ such that processor $\p_0$ evaluates all nodes defined by the associated
$1$-dimensional arrays: more formally, $\langle i_1,\dots,i_{d-1} \rangle \in N'_0$,
if $\langle i_1,\dots,i_{d-1} \rangle \in N_0$ and processor $\p_0$ evaluates
$\langle 0,i_1,\dots,i_{d-1} \rangle,\dots,\langle n-1,i_1,\dots,i_{d-1}\rangle$.
Let $n'_0 = \vert N'_0 \vert$. If $n'_0 > n_0/2^{1/d}$, we get
\[
W \geq n'_0 n > \frac{n_0 n}{2^{1/d}} \geq \frac{W^{(d-1)/d} n}{2^{1/d}} \geq W,
\]
where in the last inequality we have used the hypothesis that $W \leq n^d/2$.
This rises a contradiction, and thus it must be that $n'_0 \leq n_0/2^{1/d}$.
Then, there are $n_0 - n'_0 \geq (1-1/2^{1/d}) W^{(d-1)/d}$ points in $N_0$
whose respective $1$-dimensional arrays have not been completely evaluated by
$\p_0$: therefore, there is one communication arc associated to each array,
and the lemma follows. 

If $W > n^d/2$, we consider the remaining $p-1$ processors as a single virtual
processor evaluating $n^d-W \leq n^d/2$ nodes, being recomputation disallowed.
An argument equivalent to the previous one gives the claim.
\end{proof}

Now we need a second lemma that bounds from below the number of messages exchanged
by a processor $\p_\ell$ while evaluating nodes in an $\ell$-owned block and
in an adjacent block which is not $\ell$-owned. 

\begin{lemma}\label{lem:owned}
Consider an $\ell$-owned block $B$ adjacent to a shared or $\ell'$-owned block
$B'$, with $\ell \neq \ell'$. Then, the number of messages exchanged by
processor $\p_\ell$ for evaluating, without recomputation, nodes in $B$ and $B'$
is
\[
\BOM{\frac{n^{d-1}}{p}}.
\]
\end{lemma}

\begin{proof}
We suppose without loss of generality that $B = B_{0,0,\dots,0}$ and
$B' = B_{1,0,\dots,0}$. We call a node \emph{blue} if it is evaluated
by $\p_\ell$, and \emph{red} otherwise; let $n_{b}$ and $n_{r}$ (resp.,
$n'_{b}$ and $n'_{r}$) be the number of blue and red nodes in $B$ (resp.,
$B'$). By definition, $n_b\geq n^d/(2p^{d/(d-1)})$. Moreover, since $B'$
is either shared or $\ell'$-owned with $\ell'\neq \ell$, we have that
$n'_r \geq n^d/(2p^{d/(d-1)})$. 

Suppose $n_b \geq 3n^d/(4p^{d/(d-1)})$ and $n'_r\geq 3n^d/(4p^{d/(d-1)})$.
Consider the $n^{d-1}/p$ arrays defined by nodes in $B$
and $B'$ sharing the last $d-1$ indexes (i.e., the lines orthogonal to the
adjacency face): more formally, for each $0 \leq i_1,\dots,i_{d-1} < n/p^{1/(d-1)}$,
each array contains nodes $\langle 0,i_1,\dots,i_{d-1}, \rangle, \dots, \langle
2n/p^{1/(d-1)},i_1,\dots,i_{d-1} \rangle$. There are $2n^{d-1}/(3p)$ arrays
containing at least $n/(4p^{1/(d-1)})$ blue (resp., red) nodes in $B$ (resp.,
$B'$). Then, at least  $n^{d-1}/(3p)$ arrays contain both red and blue nodes,
and thus for each of them there is a communication arc. Since recomputation
is disallowed, each of these arcs entails the communication of a datum, and
the lemma follows.

Finally, suppose $n_b < 3n^d/(4p^{d/(d-1)})$ (resp., $n'_r < 3n^d/(4p^{d/(d-1)})$).
The lemma follows by applying Lemma~\ref{lem:w} to $B$ (resp., $B'$), with
$W = n_b$ (resp., $W = n'_r$), and considering processors $\p_i$ with
$i \neq \ell$ as a single virtual processor.
\end{proof}

The next theorem gives the claimed $\BOM{{n^{d-1}}/({p^{(d-2)/(d-1)}}) +
W^{(d-1)/d}}$ lower bound, and its proof is inspired by the argument
in~\cite{Tiskin98} for the cube DAG (which however assumes balanced work).
The lower bound is matched by the balanced algorithm given in~\cite{Tiskin98},
which decomposes the $(n,d)$-array into $p$ arrays with dimension $d$ and
size $n/p^{1/(d-1)}$. 
% We observe that our lower bound shows that the communication
% complexity is minimized by work-balanced algorithms, that is, when $W = \BT{n^d/p}$.

\begin{theorem}\label{th:cube}
Let $\A_d$ be any algorithm for solving the $(n,d)$-array problem, without recomputation,
on a BSP with $p$ processors, where $1 < p \leq n^{d-1}$, and let $W$ be the maximum
number of nodes evaluated by a processor. If $W \leq \epsilon n^d$, for an arbitrary
constant $\epsilon \in (0,1)$, then the communication complexity of the algorithm is 
\[
H_{\A_d}(n,p) = \BOM{\frac{n^{d-1}}{p^{(d-2)/(d-1)}} + W^{(d-1)/d}}.
\]
\end{theorem}

\begin{proof}
The second term of the lower bound follows directly from Lemma~\ref{lem:w} and dominates
the first one as long as $W = \BOM{n^d/(p^{d(d-2)/(d-1)^2})}$. In the remaining, we
focus on the first one and  assume $W < n^d/(p^{d(d-2)/(d-1)^2})$.

Suppose the number of shared blocks to be at least $p^{d/(d-1)}/2$. Then,
there exists a processor, say $\p_0$, computing $W' \geq n^d/(2p)$ nodes in
$b\geq 1$ shared blocks. Denote with $w_i$, for $0 \leq i < b$, the number of nodes
computed by $\p_0$ in the $i$-th shared block. We have $\sum_{i=0}^{b-1} w_i = W'$.
By Lemma~\ref{lem:w}, the messages exchanged by $\p_0$ within the $i$-th block
are $\BOM{w^{(d-1)/d}_i}$, and thus the communication complexity is at least
\[
H_{\A_d}(n,p) = \BOM{\sum_{i=0}^{b-1} w_i^{(d-1)/d}}.
\]
The summation is minimized when each $w_i$, $i \in \{0,1,\dots,b-1\}$, is set
to the maximum allowed value, that is $w_i = n^d/(2p^{d/(d-1)})$ since each block is
shared, and $b$ is set to $W	'/(n^d/(2p^{d/(d-1)})) \geq {p^{1/(d-1)}}$. The claim follows. 

Suppose now the number of shared blocks to be less than $p^{d/(d-1)}/2$.
Intuitively, in the following argument we search for an hypercube that is
almost entirely evaluated by a single processor communicating an amount of
messages proportional to the surface area. Then, we highlight a critical
sequence of these hypercubes which are evaluated one after the other: the
total amount of communication performed by the associated processors gives
the claimed bound.

We define a \emph{chain} of length $f$ a sequence of blocks
$B_{i^0_0,\dots,i^0_{d-1}},\dots,B_{i^f_0,\dots,i^{f}_{d-1}}$ such that
$i_k^{j}=i_k^{j-1}+1$, for each $0 \leq k < d$ and $1 \leq j < f$. For
instance, when $d=3$ a chain is a sequence of blocks parallel to the main
diagonal. Since less than $p^{d/(d-1)}/2$ blocks are shared, there exists
at least one chain of length $c'p^{1/(d-1)}$ containing at least $cp^{1/(d-1)}$
owned blocks, for suitable constants $c'$ and $c$ in $(0,1)$. For simplicity,
we denote with $B_k$ the $k$-th block in this chain, for $0 \leq k < c' p^{1/(d-1)}$.

Consider an $\ell$-owned block $B_k$. An \emph{$s$-hypercube} of $B_k =
B_{i^k_0,\dots,i^{k}_{d-1}}$ is defined as the set of blocks
$\cup_{0 \leq i_0,\ldots, i_{d-1} < s} \{B_{i^k_0+i_0,\dots,i^{k}_{d-1}+i_{d-1}}\}$,
and an $s$-hypercube is $\ell$-owned if at least $s^d/2$ blocks are $\ell$-owned.
Note that an $s$-hypercube of $B_k$ contains $B_k,B_{k+1},\dots,B_{k+s-1}$.
The $k$-size $s_k$ of $B_k$ is the smallest value such that the $s_{k}$-hypercube
is not $\ell$-owned, that is, there are at least $s_{k}^d/2$ shared or
$\ell'$-owned blocks, with $\ell' \neq \ell$. By definition, an $s_{k}$-hypercube
contains at least $(s_{k}-1)^d/2$ $\ell$-owned blocks. Envision an $s_k$-hypercube
as a $d$-dimensional array of size $s_k$ where processor $\p_{\ell}$ evaluates
$(s_{k}-1)^d/2 \leq W' < s_k^d/2$ nodes (i.e., the $\ell$-owned blocks);
by Lemma~\ref{lem:w}, there are $\BOM{W'^{(d-1)/d}}$ communication arcs: that is,
there are $\BT{s_k^{d-1}}$ $\ell$-owned blocks adjacent to (distinct) non
$\ell$-owned blocks. Then, by Lemma~\ref{lem:owned}, $\p_\ell$ exchanges
$\BOM{s_k^{d-1} n^{d-1}/p}$ messages for evaluating the nodes of the $s_k$-hypercube.

Consider now the longest sequence of owned blocks $B_{k_0},\dots,B_{k_{t-1}}$
such that $k_i-k_{i-1} \geq s_{i-1}$ for each $0 < i < t$, and $k_{t-1}\leq (c'
p^{1/(d-1)}-2p^{1/(d-1)^2})$. We observe that the last assumption guarantees that
the hypercube in $B_{k_{t-1}}$  with size $s_{k_{t-1}}$ is well defined, that is, it does not
exceed the boundaries of the $(n,d)$-array: indeed, since $W<n^d/(p^{d(d-2)/(d-1)^2})$, the side
$s_{k_{t-1}}$ is smaller than $2p^{1/(d-1)^2}$. We also notice that  blocks
$B_{k_0},\dots,B_{k_{t-1}}$ may be owned by a different processor. By construction, each node
within block $B_{k_i}$ depends on all the nodes in $B_{k_{i-1}}$, and thus all
messages exchanged while evaluating nodes in $B_{k_i}$ are subsequent to those
exchanged while evaluating nodes in $B_{k_{i-1}}$. Then, by summing the amount
of messages exchanged by the owners of the $t$ blocks for evaluating nodes in
the respective $s_{k_i}$-hypercubes, we have
\begin{equation}\label{eq:lbsten}
H_{\A_d}(n,p) = \BOM{\sum_{i=0}^{t-1} \frac{s_{k_i}^{d-1} n^{d-1}}{p}}.
\end{equation}
Let $S = \sum_{i=0}^{t-1} s_{k_i}$. Since there are $\BT{p^{1/(d-1)}}$
owned blocks in the chain, we have $S=\BT{p^{1/(d-1)}}$.
Equation~\ref{eq:lbsten} is minimized by setting $s_{k_i} = S/t$, and then
we obtain the claimed communication complexity by exploiting the fact that
$t \leq p^{1/(d-1)}$.
\end{proof}

\section{Sorting}\label{sec:sorting}

In this section we give a lower bound to the communication complexity of
comparison-based sorting algorithms. Comparison sorting is defined as the problem
in which a given set $X$ of $n$ input keys from an ordered set has to be sorted,
such that the only operations allowed on members of $X$ are pairwise comparisons.
Our bound only requires that no processor does more than a constant fraction $\e$
of the $\BT{n \log n}$ comparisons required by any comparison sorting algorithm,
for any $\e \in (0,1)$, and does not impose any protocol on the distribution of
the inputs and the outputs on the processors, nor upper bounds to the size of
their local memories, or specific communication patterns. As for previous work,
we still need the technical assumptions that the inputs are not initially replicated,
and that the processors store only a constant number of copies of any input key
at any moment during the execution of the algorithm.

The main result follows from the application of two lemmas, each of which provides
a different and independent lower bound to the communication complexity of sorting.
Both rely on non-trivial counting arguments, adapted from~\cite{AggarwalV88,AggarwalCS90},
that hinge on the fact that any comparison sorting algorithm must be able to distinguish
between all the $n!$ permutations of the $n$ inputs. The first lemma provides a lower
bound as a function of the maximum number $S$ of input keys initially held by a processor.
The second gives a lower bound as a function of the number $\Pi$ of permutations that
can be distinguished before any communications take place. We begin by stating and
proving the first lemma.
 
\begin{lemma}\label{lem:sorting}
Let ${\mathcal A}$ be any algorithm sorting $n$ keys on a BSP with $p$ processors,
with $1 < p \leq n$, and let $S$ denotes the maximum number of input keys initially
held by a processor. If each processor performs at most $\e(n \log n)$ comparisons,
with $\e$ being an arbitrary constant in $(0,1)$, and the input is not initially
replicated, then the communication complexity of the algorithm is
\[
H_{\A}(n,p) = \BOM{S}.
\]
\end{lemma}

\begin{proof}
Without loss of generality, denote with $\p_0$ a processor holding $S$ input
keys at the beginning, and let $\p^*$ identify the submachine including the
remaining $p-1$ processing units. Clearly, since the input is not initially
replicated, $\p^*$ initially holds $n-S$ input elements. Finally, for convenience,
we redefine $\e$ as $1/(1+\de)$, with $\delta$ being an arbitrary constant greater
than zero.

Suppose first that $S > \left(1-\frac{\de}{4(1+\de/2)}\right)n$. By hypothesis,
each processor performs at most $(n\log n)/(1+\de)$ comparisons and thus 
processor $\p_0$ can boost the number of distinguishable permutations by a factor
of at most 
\[
2^{\frac{n \log n}{1+\de}} \leq \left(\frac{n}{e(1+\de/2)}\right)^{\frac{n}{(1+\de/2)}}
\leq \left(\frac{n}{1+\de/2}\right)!,
\]
where the first inequality can be verified by taking the logarithm of both sides,
and applies for $n$ larger than a suitable constant, while the second one follows
from Stirling's approximation. This holds independently of the number of keys that
$\p_0$ contains initially (which could be even $n$) or that it receives by $\p^*$
during the execution of the algorithm. Therefore, $\p^*$ must distinguish at least
$n!/(n/(1+\de/2))!$ permutations. Then, if we denote with $S^*=n-S$ the number of
keys initially held by $\p^*$, and with $h^*$ the number of keys sent by $\p_0$ to
$\p^*$, we must have
\[
(S^*+h^*)! \geq \frac{n!}{(n/(1+\de/2))!}.
\]
By taking the logarithm of both sides and after some manipulation, we obtain
\[
(S^*+h^*) \log (S^*+h^*) \geq \frac{\de n}{2(1+\de/2)} \log \frac{n}{1+\de/2},
\]
from which follows
\[
S^*+h^* \geq \frac{\de n}{3(1+\de/2)}.
\]
Then, since $S^* < \frac{\de n}{4(1+\de/2)}$ and $S \leq n$,
\[
h^* > \frac{\de n}{12(1+\de/2)} \geq \frac{\de S}{12(1+\de/2)},
\]
and the lemma follows.

Now consider the case $S \leq \left(1-\frac{\de}{4(1+\de/2)}\right)n$. Let $h'$ and
$h^*$ be the number of keys received by $\p_0$ and $\p^*$, respectively, and let
$V'$ and $V^*$ be the maximum number of permutations distinguished by $\p_0$ and $\p^*$,
respectively. We must have $V'V^* \geq n!$. We also have $V' \leq S! \binom{S+h'}{h'}$:
indeed, $\p_0$ can distinguish all the $S!$ permutations of the $S$ input keys, and the
number of ways to intersperse the $h'$ received keys within the group of $S$ inputs
is $\binom{S+h'}{h'}$. (Note that the $h'!$ permutations of the $h'$ messages are
accounted in $V^*$.) Similarly, $V^* \leq (n-S)! \binom{n-S+h^*}{h^*}$. Thus, we have
\[
S! (n-S)! \binom{S+h'}{h'}\binom{n-S+h^*}{h^*} \geq n!,
\]
whence
\[
\binom{S+h}{h}\binom{n-S+h}{h} \geq \binom{n}{S},
\]
where $h = \max\{h',h^*\}$. By using the fact that $(a/b)^b \leq \binom{a}{b} \leq (e a/b)^b$
for any integer values $a$ and $b$, and then by taking the logarithm of both sides, we get
\begin{equation}\label{eq:sort_log}
h \log\left(\frac{e^2}{h^2} (S+h)(n-S+h)\right) \geq S \log \frac{n}{S},
\end{equation}
where $e$ is Euler's constant.

In the rest of the proof we will prove that $h \geq \beta S$ for a suitable constant
$\beta \in (0,1)$ that will be defined later. Suppose, for the sake of contradiction,
that $h < \beta S$. We first observe that the left-hand side of Equation~\ref{eq:sort_log}
is increasing in $h$. Indeed, we have 
\[
h \log\left(\frac{e^2}{h^2} (S+h)(n-S+h)\right) = 
2h \log e + h \log \frac{S+h}{h} + h\log \frac{n-S+h}{h},
\]
where $h \log ((x+h)/h)$, with $x \in \{S, n-S\}$, is strictly increasing in $h$
as soon as $x > 0$. Therefore, the left-hand side of Equation~\ref{eq:sort_log}
can be upper bounded as follows:
\begin{align*}
h \log\left(\frac{e^2}{h^2} (S+h)(n-S+h)\right)
&< \beta S \log\left(e^2 \frac{(S + \beta S)}{\beta S} \frac{(n - S + \beta S)}{\beta S}\right) \\
&= \beta S \log\left(\frac{e^2 (1 + \beta) (n - S(1 - \beta))}{\beta^2 S}\right) \\
&< \beta S \log\left(\frac{2 e^2 n}{\beta^2 S}\right) \\
&< S \log\left(\frac{2en}{\beta S}\right)^{2\beta},
\end{align*}
where we have also used the facts that $\beta < 1$ and $S \leq n$. We now argue that
the last term in the above formula is upper bounded by $S \log(n/S)$. We shall
consider two separate cases. The first is when $n/S \geq 2$. In this case, we
set $\beta = \log(n/S)/(8 \log(2en/S))$. (Observe that $0 < \beta < 1$, as required.)
Standard calculus shows that $8 \log(2en/S)/\log(n/S) < (2en/S)^3$ when $n/S \geq 2$.
Hence, we can write
\begin{align*}
\left(\frac{2en}{\beta S}\right)^{2\beta}
&= \left(\frac{2en \cdot 8 \log(2e n/S)}{S \log(n/S)}\right)^{\frac{\log (n/S)}{4 \log(2e n/S)}} \\
&< \left(\frac{2e n}{S}\right)^{\frac{\log(n/S)}{\log(2e n/S)}} \\
&= \frac{n}{S}.
\end{align*}
Consider now the case $4(1+\de/2)/(4+\de) \leq n/S < 2$. Then, we have
$(2en/(\beta S))^{2\beta} < (11/\beta)^{2\beta}$. Since $\delta$ is a constant,
and since the right-hand term of the above inequality tends to one as $\beta$
tends to zero, then for each $\delta > 0$ there exists a constant $\beta \in (0,1)$
such that $(11/\beta)^{2\beta} \leq 4(1+\de/2)/(4+\de)$. Therefore, we have shown
for both cases that, if $h < \beta S$,
\[
h \log\left(\frac{e^2}{h^2} (S+h)(n-S+h)\right) < S \log\frac{n}{S},
\]
which is in contradiction with Equation~\ref{eq:sort_log}. It follows that
there exists a constant $\beta > 0$ such that $h \geq \beta S$, giving the lemma.
\end{proof}

We now provide a second lemma, which bounds from below the communication complexity
of sorting in BSP as a function of the number $\Pi$ of permutations that can be
distinguished before any communications take place, that is, when processors'
can only compare their local inputs.

As an aside, we observe that the proof of this lemma can be straightforwardly cast for
the LPRAM model, yielding a much simpler proof for Theorem 3.2 of~\cite{AggarwalCS90},
which bounds from below the \emph{communication delay}, that is, the number of
communication steps, required for comparison-based sorting.

\begin{lemma}\label{lem:permutations}
Let ${\mathcal A}$ be any algorithm sorting $n$ keys on a BSP with $p$ processors,
where $1 < p \leq n$, and let $\Pi$ be the number of distinct permutations that can
be distinguished by ${\mathcal A}$ \emph{before} the second superstep, that is,
by comparing the inputs that (possibly) reside initially in the processors' local
memories. If ${\mathcal A}$ stores only a constant number of copies of any key at
any time instant, then the communication complexity of the algorithm is 
\[
H_{\A}(n,p) = \BOM{\frac{n \log n - \log \Pi}{p \log (n/p)}}.
\]
\end{lemma}

\begin{proof}
We prove the lemma only for the case when every input key is present in only
one of the local memories of the processors at any time instant; the extension
to the case when a data element is simultaneously present in a constant number
of local memories is straightforward and thus omitted.

We suppose that $\A$ performs $1$-relations in each superstep, that is,
each processor can send and receive only one message. This is without loss
of generality because we observe that each superstep of $\A$ where each
processor performs an $h$-relation (i.e., it sends and receives at most
$h$ messages), can be decomposed into $h$ $1$-relation supersteps without
increasing the communication complexity of $\A$ (since the latter does not
charge a synchronization cost due to the latency incurred by each superstep).
Let $m_j$ denote the number of input keys in local memory of processor $\p_j$
after a given superstep of the algorithm. Since, by hypothesis, a data element
is present in only one of the local memories of the processors at any time
instant, we have that $\sum_{j=1}^{p} m_j \leq n$. Hence, after a communication
superstep, which by hypothesis entails a $1$-relation, the space of permutations
can be divided, at most, by the value of an optimal solution of the following
convex program (observe, in fact, that $\p_j$ may already have distinguished
$(m_j - 1)!$ permutations before the last superstep): 
\begin{align*}
\max &\, \prod_{j=0}^{p-1} m_j\\
\text{s.t.} &\, \sum_{j=0}^{p-1} m_j \leq n.
\end{align*}
Since the solution is given by $m_j = n/p$ for each $j$, its value is $(n/p)^p$.
Thus, after $x$ supersteps, the space of permutations can have been divided by
at most $(n/p)^{px}$. Since there remain $n!/\Pi$ distinct possible permutations,
we must have
\[
\left(\frac{n}{p}\right)^{px} \geq \frac{n!}{\Pi}.
\]
By taking the logarithm of both sides, we obtain
\[
x = \BOM{\frac{n \log n - \log \Pi}{p \log (n/p)}},
\]
as desired.
\end{proof}

Now we are ready to prove the main result of this section, an $\BOM{(n \log n)/(p \log(n/p))}$
lower bound to the communication complexity of any comparison sorting algorithm.
The result follows by combining the bounds given by the previous two lemmas. Both
bounds are not tight when considered independently, the first (Lemma~\ref{lem:sorting})
because it is weak when at the beginning the input keys tend to be distributed evenly
among the processors, the second (Lemma~\ref{lem:permutations}) because it is
weak when the input keys tend to be concentrated on one or few processors. However,
the simultaneous application of both provides the sought (tight) lower bound.

\begin{theorem}\label{thm:sorting}
Let ${\mathcal A}$ be any algorithm for sorting $n$ keys on a BSP with $p$
processors, with $1 < p \leq n$. If each processor performs at most $\e (n \log n)$
comparisons, with $\e$ being an arbitrary constant in $(0,1)$, the inputs are not
initially replicated, and the $p$ processors store only a constant number of copies
of any key at any time instant, then the communication complexity of the algorithm
is
\[
H_{\A}(n,p) = \BOM{\frac{n \log n}{p \log(n/p)}}.
\]
\end{theorem}

\begin{proof}
The result follows by combining Lemma~\ref{lem:sorting} with Lemma~\ref{lem:permutations}.
Since, by hypothesis, each processor performs at most $\e (n \log n)$ comparisons,
with $\e \in (0,1)$, and the inputs are not initially replicated, we can apply
Lemma~\ref{lem:sorting}, obtaining
\[
H_{\A}(n,p) = \BOM{S},
\]
where $S$ denotes the maximum number of input keys initially held by a processor.
Moreover, since by hypothesis the $p$ processors store a constant number of copies
of any key at any time instant, we can also apply Lemma~\ref{lem:permutations},
obtaining
\begin{equation}\label{eq:lemma_permutations}
H_{\A}(n,p) = \BOM{\frac{n \log n - \log \Pi}{p \log (n/p)}},
\end{equation}
where $\Pi$ denotes the number of distinct permutations that can be distinguished
by $\A$ by comparing the inputs that initially reside in processors' local memories.
In order to compare the latter bound with the first one, we need to bound $\Pi$ from
above as a function of $S$. To this end, let $s_i$ denote the number of input keys
initially held by processor $\p_i$. Hence, $S = \max \{s_0,s_1,\dots,s_{p-1}\}$.
The number of permutations that can be distinguished by ${\mathcal A}$ without
requiring communication, that is, by letting each processor sort the keys that it
holds at the beginning of the computation, is therefore $\Pi = \prod_{i=0}^{p-1} s_i!$.
Since the inputs are not initially replicated, an upper bound to $\Pi$ as a function
of $S$ is given by the value of an optimal solution of the following mathematical
program:
\begin{align*}
\max &\, \prod_{i=0}^{p-1} s_i!\\
\text{s.t.} &\, \sum_{i=0}^{p-1} s_i = n\\
&\quad\quad s_i \leq S \quad \forall i = 0,1,\dots,p-1.
\end{align*}
Since $a!b! \leq (a+b)!$ for any integer $a$ and $b$, by a convexity argument
it follows that $\prod_{i=0}^{p-1} s_i! \leq (S!)^{n/S}$. Therefore, we can plug
$\Pi = (S!)^{n/S}$ in Equation~\ref{eq:lemma_permutations}, obtaining
\[
H_{\A}(n,p) = \BOM{\frac{n \log n - n \log S}{p \log(n/p)}}.
\]
Putting pieces together, we conclude that
\[
H_{\A}(n,p) = \BOM{\frac{n \log(n/S)}{p \log(n/p)} + S}.
\]
Standard calculus shows that the right-hand side of the above equation is increasing
in $S$ when $S = \BOM{n/(p \log(n/p))}$. The theorem follows by observing that
$S \geq \lceil n/p \rceil$.
\end{proof}

\section{Fast Fourier Transform}\label{sec:FFT}

In this section we consider the problem of computing the Discrete Fourier
Transform of $n$ values using the $n$-input FFT DAG. In the FFT DAG, a vertex
is a pair $\langle w,l \rangle$, with $0 \leq w < n$ and $0 \leq l \leq \log n$,
and there exists an arc between two vertices $\langle w,l \rangle$ and
$\langle w',l' \rangle$ if $l' = l+1$, and either $w$ and $w'$ are identical
or their binary representations differ exactly in the $l'$-th bit. We show that,
when no processor computes more than a constant fraction of the total number of
vertices of the DAG, the communication complexity is $\BOM{n \log n/(p \log(n/p))}$.
Our bound does not assume any particular I/O protocol, and only requires that
every input resides in the local memory of exactly one processor before the
computation begins; as for preceding results, our bound also hinges on the
restriction on the nature of the computation whereby each vertex of the FFT DAG
is computed exactly once. The bound is tight for any $p \leq n$, and is achieved
by the well-known recursive decomposition of the DAG into two sets of smaller
$\sqrt{n}$-input FFT DAGs, with each set containing $\sqrt{n}$ of such
subDAGs (see, e.g.,~\cite{BilardiPPS07}).

We will first establish a lemma which, under the same hypothesis of the main
result, provides a lower bound to the communication complexity as a function
of the maximum work performed by any processor. The proof of the lemma is based
on a bandwidth argument, which exploits the fact that an FFT DAG can perform all
cyclic shifts (see, e.g.,~\cite{Leighton92}), and on the following technical result
which is implicit in the work of Hong and Kung (a simplified proof is due to
Aggarwal and Vitter~\cite{AggarwalV88}).

\begin{lemma}[{\cite{HongK81}}]\label{lem:HK-FFT}
Consider the computation of the $n$-input FFT DAG. During the computation,
if a processor accesses at most $S$ nodes of the DAG, then it can evaluate
at most $2S \log S$ nodes, for any $S \geq 2$.
\end{lemma}

\begin{lemma}\label{lem:cyclic_shifts}
Let ${\mathcal A}$ be any algorithm computing, without recomputation, an $n$-input
FFT DAG on a BSP with $p$ processors, with $1 < p \leq n$, and let $W$ be the maximum
number of nodes of the FFT DAG computed by a processor. If $W \leq \epsilon (n \log n)$,
for an arbitrary constant in $(0,1)$, and the inputs are not initially replicated,
then the communication complexity of the algorithm is
\[
H_{\A}(n,p) = \BOM{\frac{W}{\log W}}.
\]
\end{lemma}

\begin{proof}
Let $\p_0$ be a processor computing $W$ nodes of the FFT DAG, and consider as an
unique processor $\p^*$ the remaining $p-1$ processing units.

Suppose first that processor $\p^*$ contains at least $n/2$ of the $n$ output nodes
at the end of the algorithm. Let $K = W/(2 \log W)$. Since $\p_0$ evaluates $W$
nodes, it follows from Lemma~\ref{lem:HK-FFT} that $\p_0$ accesses at least $K$
node values during the execution of the algorithm. These nodes can be either inputs
initially held by $\p_0$, or nodes whose values have been evaluated and then sent
by processor $\p^*$. If at least $K/2$ of them have been sent by $\p^*$, the
lemma follows. Otherwise, $\p_0$ initially contains at least $K/2$ input nodes.
Since the inputs are not initially replicated, and since an FFT DAG can perform
all cyclic shifts, by~\cite[Lemma 10.5.2]{Savage98} there exists a cyclic shift
that permutes $K/4$ input nodes initially held by processor $\p_0$ into $K/4$
output nodes held by $\p^*$ at the end of the algorithm. Since $K = W/(2 \log W)$
and since, by hypothesis, $W \leq \epsilon (n \log n)$, it holds that
$K/4 \leq n/2$, and thus at least $K/4$ messages are actually needed.
Therefore, $H_{\A}(n,p) \geq W/(8 \log W)$.

Now suppose that processor $\p^*$ contains at most $n/2$ output nodes at the
end of the algorithm. Thus, there are at least $n/2$ output nodes in $\p_0$.
Since, by hypothesis, recomputation is disallowed, $\p^*$ computes
$W^* = n \log n - W \leq n \log n$ nodes of the DAG. The lemma follows
by inverting the role of $\p_0$ and $\p^*$ and setting $K = W^*/(2 \log W^*)$
in the previous argument.
\end{proof}

We note that the above bound is matched when $W = \BO{n^\e \log n}$, for any
constant $\e \in (0,1)$, by the previous recursive algorithm by ending the
recursion when the subproblem size is $\BT{W/\log W}$.

The main result of this section follows by a simple application of the preceding
lemma and of a result implicit in the proof of the lower bound due to Bilardi et
al.~\cite[Corollary~1]{BilardiSS12}.

\begin{theorem}\label{thm:fft-lb}
Let ${\mathcal A}$ be any algorithm computing, without recomputation, an $n$-input FFT
DAG on a BSP with $p$ processors, where $1 < p \leq n$. If each processor computes at most
$\epsilon (n \log n)$ nodes, for an arbitrary constant in $(0,1)$, of the FFT DAG and the
inputs are not initially replicated, then the communication complexity of the algorithm is
\[
H_{\A}(n,p) = \BOM{\frac{n \log n}{p \log(n/p)}}.
\]
\end{theorem}

\begin{proof}
If $W \geq n^{1/4}$, we have that $W \geq \max\{(n \log n)/p, n^{1/4}\}$, and thus
we observe that the bound given by Lemma~\ref{lem:cyclic_shifts} dominates the one
claimed by the theorem. Otherwise, when $W < n^{1/4}$, we use the following argument.
By reasoning as in~\cite[Corollary~1]{BilardiSS12}, if at the end of the algorithm
${\mathcal A}$ each processor holds at most $U \leq n$ output nodes of the FFT
DAG and recomputation is not allowed, then the communication complexity of
${\mathcal A}$ is $\BOM{\max\{0, n \log(n/U^2)/(p\log (n/p))\}}$. Since
$W < n^{1/4}$, each processor cannot contain more than $n^{1/4}$ output nodes,
that is, $U \leq n^{1/4}$, and the theorem follows.
\end{proof}

\section{Conclusions}\label{sec:conclusions}

We have presented new lower bounds on the amount of communication required
to solve some key computational problems in distributed-memory parallel
architectures. All our bounds have the same functional form of previous
results that appear in the literature; however, the latter are built by
making a critical use of some assumptions that rule out a large part of
possible algorithms. The novelty and the significance of our results stem
from the assumptions under which our lower bounds are developed, which
are much weaker than those used in previous work.

Our bounds are derived within the BSP model of computation, but can be
easily extended to other models for distributed computations based on
or similar to the BSP, such as LogP~\cite{CullerKPSSSSV96} and 
MapReduce~\cite{KarloffSV10,PietracaprinaPRSU12}. Moreover, we believe
that our results can be also ported to models for multicore computing 
(see, e.g.,~\cite{BlellochCGRCK08,Valiant11,ChowdhuryRSB13}), since our
proofs are based on some techniques that have already been exploited
in this scenario.

There is still much to do towards the establishment of a definitive theory
of communication-efficient algorithms. In fact, we were not able to remove
all the restrictions there were in place in previous work: in some cases
our lower bounds still make use of some technical assumptions, such as the
non-recomputation of intermediate results, or restrictions on the replication
of input data. Although it seems that such restrictions can be relaxed to
encompass a small amount of recomputation or input replication, it is an
open question to assess whether these assumptions are inherent to our proof
techniques or can be removed. In particular, it is not clear, in general,
when recomputation has the power to reduce communications, since many
lower bound techniques do not apply in this more general scenario (see,
e.g.,~\cite{BallardDHS12}). Providing tight lower bounds that hold also
when recomputation is allowed is a fascinating and challenging avenue for
future research.

\paragraph*{Acknowledgments.}
The authors would like to thank Gianfranco Bilardi and Andrea Pietracaprina
for useful discussions.

\bibliographystyle{abbrv}
\bibliography{biblio}

\end{document}